\documentclass[letterpaper, 10pt, journal, twoside]{IEEEtran}
\usepackage{amsmath,amsfonts}
\usepackage{algorithm}
\usepackage{algpseudocode}
\usepackage{amsmath} 
\usepackage{array}
\usepackage{subfigure}
\usepackage{textcomp}
\usepackage{stfloats}
\usepackage{url}
\usepackage{verbatim}
\usepackage{graphicx}
\usepackage{cite}

\usepackage{amsthm,amssymb}
\usepackage{color}
\usepackage{pifont}
\usepackage{hyperref}
\usepackage{footnote}
\usepackage{enumitem}
\usepackage{diagbox}
\usepackage{booktabs}
\usepackage{multirow}
\usepackage{tablefootnote}
\hypersetup{
colorlinks=true,
linkcolor=red,
citecolor=green,
}

\graphicspath{{Figures/}}

\newtheorem{theorem}{Theorem}[section]

\newtheorem{lemma}[theorem]{Lemma}

\newtheorem{assumption}[theorem]{Assumption}
\theoremstyle{definition}
\newtheorem{remark}[theorem]{Remark}

\newcommand{\longthmtitle}[1]{\mbox{}{\textit{(#1):}}}

\newcommand{\real}{\ensuremath{\mathbb{R}}}
\newcommand{\complex}{\ensuremath{\mathbb{C}}}

\newcommand{\Cc}{\mathcal{C}}
\newcommand{\Dc}{\mathcal{D}}
\newcommand{\Ec}{\mathcal{E}}

\newcommand{\Lc}{\mathcal{L}}

\newcommand{\Nc}{\mathcal{N}}

\newcommand{\Qc}{\mathcal{Q}}

\newcommand\abf{\mathbf{a}}

\newcommand\gbf{\mathbf{g}}
\newcommand\hbf{\mathbf{h}}

\newcommand\pbf{\mathbf{p}}
\newcommand\qbf{\mathbf{q}}

\newcommand\vbf{\mathbf{v}}

\newcommand\ybf{\mathbf{y}}

\newcommand\Abf{\mathbf{A}}

\newcommand\Ebf{\mathbf{E}}

\newcommand\Ibf{\mathbf{I}}
\newcommand\Jbf{\mathbf{J}}

\newcommand\Rbf{\mathbf{R}}

\newcommand\Xbf{\mathbf{X}}
\newcommand\Ybf{\mathbf{Y}}

\newcommand\phib{\boldsymbol{\phi}}
\newcommand\psib{\boldsymbol{\psi}}

\newcommand{\ones}{\mathbf{1}}
\newcommand{\zeros}{\mathbf{0}}

\newcommand{\oprocendsymbol}{\hbox{$\bullet$}}
\newcommand{\oprocend}{\relax\ifmmode\else\unskip\hfill\fi\oprocendsymbol}

\renewcommand{\baselinestretch}{.99}
\allowdisplaybreaks

\makeatletter
\def\footnoterule{\kern-3\p@
  \hrule \@width 2in \kern 2.6\p@} 
\makeatother

\begin{document}

\title{\LARGE \bf Constraints on OPF Surrogates for Learning
Stable \\ Local Volt/Var Controllers}

\author{Zhenyi Yuan, Guido Cavraro, and Jorge Cort\'es
\thanks{This work was authored by the National Renewable Energy Laboratory, operated by Alliance for Sustainable Energy, LLC, for the U.S. Department of Energy (DOE) under Contract No. DE-AC36-08GO28308. Funding provided by the NREL Laboratory Directed Research and Development Program. The views expressed in the article do not necessarily represent the views of the DOE or the U.S. Government. The U.S. Government retains and the publisher, by accepting the article for publication, acknowledges that the U.S. Government retains a nonexclusive, paid-up, irrevocable, worldwide license to publish or reproduce the published form of this work, or allow others to do so, for U.S. Government purposes.}
\thanks{Z. Yuan and J. Cort\'es are with the Department of Mechanical and Aerospace Engineering, UC San Diego. {\tt \{z7yuan,cortes\}@ucsd.edu}. G. Cavraro is with the Power Systems Engineering Center, National Renewable Energy Laboratory. {\tt guido.cavraro@nrel.gov}.}
}

 \markboth{IEEE Control Systems Letters}%
 {Yuan \MakeLowercase{\textit{et al.}}: Constraints on OPF Surrogates for Learning Stable Local Volt/Var Controllers}


\maketitle

\begin{abstract}
We consider the problem of learning local Volt/Var controllers in distribution grids (DGs). Our approach starts from learning separable surrogates that take both local voltages and reactive powers as arguments and predict the reactive power setpoints that approximate optimal power flow (OPF) solutions. We propose an incremental control algorithm and identify two different sets of slope conditions on the local surrogates such that the network is collectively steered toward desired configurations asymptotically. Our results reveal the trade-offs between each set of conditions, with coupled voltage-power slope constraints allowing an arbitrary shape of surrogate functions but risking limitations on exploiting generation capabilities, and reactive power slope constraints taking full advantage of generation capabilities but constraining the shape of surrogate functions. AC power flow simulations on the IEEE 37-bus feeder illustrate their guaranteed stability properties and respective advantages in two DG scenarios.
\end{abstract}


\section{Introduction}\label{sec:intro}

The massive deployment of distributed energy resources (DERs) in DGs represents an opportunity to improve the performance of the power grid and reduce greenhouse gas emissions. Nevertheless, if not properly regulated, DERs'  power injections can pose challenges to system operations and stability.
For instance, the intermittence of renewable energy sources can cause large voltage variations~\cite{Mohammadi_2017_TPD}. Volt/Var control strategies aim to keep voltages within safe preassigned limits by commanding DERs' reactive power injections. 
Here, we propose to learn local Volt/Var controllers with improved optimality and rigorous performance guarantees.

\subsubsection*{Literature Review}
Generator reactive power outputs are classically computed in an \emph{open-loop} fashion as the solution of an OPF problem, see e.g.,~\cite{SHL:14}.
Learning-based approaches have recently been proposed to predict OPF solutions targeting scenarios of high DERs penetration and increased load variability, aiming to solve numerous OPF problem instances within a limited timeframe~\cite{XP-TZ-MC-SZ:21,MKS-SG-VK-GC-AB:20}.
Nevertheless, the aforesaid methods require \emph{complete} information, namely, power demands from loads and generation limits from generators must be exactly known.
This is not possible in most of the actual DGs, e.g., because individual loads are unlikely to announce their demand profiles in advance, and the evolving availability of small-size generators is hard to predict.

This has motivated the development of \emph{closed-loop} strategies, which compensate for the lack of information with measurements retrieved from the field.
Given the massive number of controllable devices hosted in future DGs, \emph{decentralized} approaches are often advocated for practical applications.
There are two notable classes of decentralized algorithms. First, we have \emph{distributed} algorithms where agents cooperate with peers. 
The literature has seen the recent development of distributed optimization-based feedback controllers, which steer the network towards OPF solutions by means of the cyclic repetition of \emph{sensing}, \emph{communication} of key variables with peers, \emph{computation} of power setpoints, and their \emph{application}, see e.g.,~\cite{ED-AS:18,AB-EDA:19,GQ-NL:20}.
Nevertheless, distributed strategies need a reliable real-time communication infrastructure, which is rarely present in actual  DGs.
Second, we have \emph{local} strategies, in which each agent makes decisions based only on information available locally, see e.g., the IEEE standard 1547~\cite{IEEE1547} and~\cite{XZ-MF-ZL-LC-SHL:21}.
Local schemes have intrinsic performance limitations and in general lack of optimality considerations~\cite{SB-RC-GC-SZ:19}.
To enhance the performance of local control schemes and reduce the performance gap with distributed optimization-based strategies, recent research efforts develop learning-based control frameworks for devising local Volt/Var strategies from data retrieved from DGs~\cite{WC-JL-BZ:22,YS-GQ-SL-AA-AW:22,GC-ZY-MKS-JC:22-cdc,ZY-GC-MKS-JC:22-tps}.
Although providing interesting insights on learning Volt/Var rules,
most of the existing works~\cite{SK-PA-GH:19,PST-PHG:22,XS-JQ-JZ:21}
do not assess the system stability and hence are not suitable for practical applications.
Also, in these methods, the local surrogates learned to predict OPF solutions only consider  voltage magnitudes, which limits performance as, from a local perspective, the same voltage magnitude might correspond to multiple OPF solutions~\cite{ZY-GC-MKS-JC:22-tps}.

\subsubsection*{Statement of Contributions}
This paper proposes\footnote{
Throughout the paper, $\real$ and $\complex$ denote the set of real and complex numbers, respectively. Upper and lower case boldface letters denote matrices and column vectors, respectively. 
Calligraphic symbols denote sets. 
Given a vector $\abf$ (resp., diagonal matrix $\mathbf A$), its $n$-th (diagonal) entry is denoted by $a_n$ $(A_n)$. $\Abf \succ(\succeq)~0$ denotes that matrix $\Abf$ is positive (semi-) definite, and $\Abf \prec(\preceq)~0$ denotes that matrix $\Abf$ is negative (semi-) definite. The symbol $(\cdot)^\top$ stands for transposition, and $\ones,\zeros, \Ibf$ denote vectors of all ones and zeros and identity matrix with appropriate dimensions, respectively. Operators $\Re(\cdot)$ and $\Im(\cdot)$ extract the real and imaginary parts of a complex-valued argument, and act element-wise. With a slight abuse of notation, we use $|\cdot|$ to denote the absolute value for real-valued arguments and the cardinality when the argument is a set. $\|\cdot\|$ represents the Euclidean norm. Given a matrix with real eigenvalues, $\lambda_{\max}(\cdot)$ and $\lambda_{\min}(\cdot)$ respectively represent its largest and smallest eigenvalue. An eigenvalue $\lambda$ and its associated eigenvector $\boldsymbol{\xi}$ form the eigenpair $(\lambda, \boldsymbol{\xi})$. 
The range of $\phi$ is the set of its possible output values. } a framework for devising local Volt/Var control schemes acting as local surrogates of OPF solvers.
Desirable network configurations are described by \emph{equilibrium functions}, which map local information to an approximation of the optimal generator reactive power output. An incremental control algorithm then steers the network toward the equilibria identified by the equilibrium functions.
Compared to our previous works~\cite{GC-ZY-MKS-JC:22-cdc, ZY-GC-MKS-JC:22-tps}, here the equilibrium functions depend not only on voltage magnitudes but also on reactive power injections, which has a significant impact in reducing the optimality gap of the proposed local strategies. 
Two different sets of conditions ensuring the asymptotic stability of equilibrium points for the incremental rule are provided. One requires coupling slope constraints on the functions of voltages and functions of reactive powers (referred to as CVP-SC), and another requires only slope constraints on the functions of reactive powers and decreasing functions of voltages (referred to as RP-SC).  Our results reveal the trade-offs between each set of conditions regarding the optimality gap; in particular the CVP-SC looks more suitable in DGs with relatively small-size generators, and the RP-SC looks more suitable in DGs with relatively large-size generators.

\section{DG Modeling and Problem Formulation}\label{sec:modeling}

A radial single-phase (or a balanced three-phase) DG having $N+1$ buses can be modeled by a tree graph $\mathcal G=(\mathcal N,\mathcal L)$ rooted at the substation. The nodes in $\mathcal N:=\{0,\ldots,N\}$ are associated with grid buses, and the edges in $\mathcal L$ with lines. Let $(m,n)$ be an edge in $\mathcal L$; $y_{(m,n)} \in \complex$ is its admittance.
Neglecting the shunt admittances, the bus admittance matrix $\tilde \Ybf \in \complex^{(N+1)\times(N+1)}$ is defined as   
\begin{align*}
(\tilde \Ybf)_{mn} = \begin{cases}
- y_{(m,n)} & \text{ if } (m,n) \in \Ec, m \neq n, \\
0 & \text{ if } (m,n) \notin \Ec, m \neq n, \\
\sum_{k \neq n} y_{(k,n)} & \text{ if }m = n.
\end{cases}
\end{align*}
Note that $\tilde \Ybf$ is symmetric and satisfies $\tilde \Ybf \ones = \zeros$. Also, by separating the components associated with the substation and the ones associated with the other nodes, $\tilde \Ybf$ is partitioned~as
\begin{align*}
\tilde \Ybf = \begin{bmatrix}
 y_{0}&\ybf_0^\top \\
 \ybf_0& \Ybf
 \end{bmatrix} ,
\end{align*}
with $y_{0} \in \complex$, $\ybf_0 \in \complex^{N}$, and $\Ybf \in \complex^{N \times N}$; $\Ybf$ is invertible when the network is connected~\cite{AMK-MP:17}, and we denote $\tilde \Rbf:=\Re(\Ybf^{-1})$ and $\tilde \Xbf:=\Im(\Ybf^{-1}) \in \real^{N \times N}$.

The voltage magnitude at bus $n\in \Nc$ is denoted as $v_n\in \real$.
The substation node, labeled as 0, behaves as an ideal voltage source imposing the nominal voltage of 1 p.u. 
The active and reactive power injections at bus $n\in \Nc$ are $p_n,q_n\in \real$, respectively. Powers take positive (negative) values, i.e., $p_n, q_n \geq 0$ ($p_n, q_n \leq 0$), when they are \emph{injected into} (\emph{absorbed from}) the grid.
The vectors $ \vbf,\pbf, \qbf \in \real^N$ collect the voltage magnitudes, active and reactive power injections for buses $1,2,\dots N$.

Voltages and powers are related by the nonlinear power flow equations but here we consider the linearization~\cite{ED-AS:18,GQ-NL:20,HZ-HJL:15}
 \begin{align}
 \vbf = \tilde \Rbf \pbf + \tilde \Xbf \qbf + \ones,
 \label{eq:v=Rp+Xq}
 \end{align}
because it will be useful to prove the stability properties of the proposed control algorithms. Nevertheless, the devised algorithms are tested on an exact AC power flow solver in Section~\ref{sec:sim}. Using~\eqref{eq:v=Rp+Xq}, power losses can be approximated as 
$\qbf^\top \tilde \Rbf \qbf + \pbf^\top \tilde \Rbf \pbf$.    

Assume a subset $\Cc \subseteq \Nc$ of buses host DERs, with $|\Cc| = C$.
Every DER corresponds to a smart agent provided with some computational and sensing capabilities, i.e., it can measure its voltage magnitude.
The remaining nodes constitute the set $\Lc=\Nc\setminus\Cc$ and are referred to as loads.
For convenience, we partition reactive powers and voltage magnitudes by grouping together the nodes belonging to the load and generation sets,
\begin{align*}
\qbf = \begin{bmatrix}
\qbf_\Cc^\top & \qbf_\Lc^\top
\end{bmatrix}^\top, \quad
\vbf = \begin{bmatrix}
\vbf_\Cc^\top & \vbf_\Lc^\top
\end{bmatrix}^\top.
\end{align*}
The matrices $\tilde \Rbf$ and $\tilde \Xbf$ can be partitioned as well, yielding
\begin{align*}
\tilde \Rbf = \begin{bmatrix}
\Rbf & \Rbf_\Lc \\
\Rbf_\Lc^\top & \Rbf_{\Lc\Lc}
\end{bmatrix}, \quad
\tilde \Xbf = \begin{bmatrix}
\Xbf & \Xbf_\Lc \\
\Xbf_\Lc^\top & \Xbf_{\Lc\Lc}
\end{bmatrix},
\end{align*}
with $\Rbf, \Xbf \succ 0$, cf.~\cite{HZ-HJL:15}.
Fixing the uncontrollable variables $\pbf, \qbf_\Lc$, and using~\eqref{eq:v=Rp+Xq}, voltage magnitudes become functions exclusively of $\qbf_\Cc$,
\begin{align}
 \label{eq:volt-approx}
 & \vbf(\qbf_\Cc) =
\begin{bmatrix}
 \Xbf \\
 \Xbf_\Lc^\top
 \end{bmatrix} \qbf_\Cc + \hat \vbf,
 \end{align}
where
 \begin{align*}
\hat \vbf = \begin{bmatrix}
\hat \vbf_\Cc\\
\hat \vbf_\Lc
\end{bmatrix} = 
\begin{bmatrix}
\Xbf_\Lc\\
\Xbf_{\Lc\Lc}
\end{bmatrix} \qbf_\Lc \!+\! \tilde\Rbf \pbf \!+\! \ones.
 \end{align*}

A distribution system operator seeks the generator reactive power injections to be optimal, i.e., to implement the solution of an \emph{optimal reactive power flow} (ORPF) problem. Here, we consider the following ORPF problem formulation, though other versions could be considered as well~\cite{ZY-GC-MKS-JC:22-tps}
\begin{subequations}\label{eq:ORPF}
	\begin{align}
	\qbf_\Cc^\star(\pbf,\qbf_\Lc):=&  \arg\min_{\qbf_\Cc}\ ~\qbf^\top \tilde \Rbf \qbf + \pbf^\top \tilde \Rbf \pbf \label{eq:ORPF:cost}\\
	\mathrm{s.t.}\  & ~ \eqref{eq:volt-approx}\notag\\
	&~\vbf_{\min} \leq \vbf(\qbf_\Cc) \leq \vbf_{\max} \label{eq:ORPF:c1}\\
    &~q_n \in \mathcal Q_n, n \in \Cc
	\end{align} 
\end{subequations}
where $\vbf_{\min}, \vbf_{\max} \in \real^N$ are desired bus voltage limits. The reactive power injections of bus $n$ must be within the feasible limits described by the set $\mathcal Q_n = \{q_n: q_{\min,n} \leq q_n \leq q_{\max,n} \}$ with $\qbf_{\min}, \qbf_{\max} \in \real^C$ being the minimum and maximum DERs' reactive power injection vectors. The cost encodes the goal of minimizing line losses.
Solving~\eqref{eq:ORPF} inevitably requires knowledge of the network-wide quantities $(\pbf,\qbf_\Lc)$,
making purely local control strategies in general unsuccessful.

To address this, we propose to obtain local surrogates of $\qbf_\Cc^\star$ by learning, for each agent $n \in \mathcal C$, a function $h_n$,
\begin{equation}
\label{eq:hn}
h_n: \Qc_n \times \real \rightarrow \Qc_n, \; (q_n,v_n) \mapsto h_n(q_n,v_n)
\end{equation}
that takes as input the current local voltage $v_n$ and reactive power $q_n$, and gives as output an approximation of the reactive power that the DER at node $n$ would inject to steer the network to the solution of~\eqref{eq:ORPF}. Note that each function $h_n$ only depends on local variables. 
Concurrently, we design local control rules whose equilibrium satisfies
\begin{equation}
q_n = h_n(q_n,v_n)
\label{eq:equilibria}
\end{equation}
and steer the network toward desired configuration described by $\{h_n\}_{n \in \Cc}$. For this reason, we refer to $\{h_n\}_{n \in \Cc}$ as \emph{equilibrium} functions. Throughout the paper, we consider equilibrium functions that meet the following assumption.

\begin{assumption}\longthmtitle{Separable and differentiable equilibrium functions}\label{ass:differentiability}
The equilibrium functions have the form
\begin{align}\label{eq:eq_fun} 
    h_n (q_n,v_n) =  \psi_n(q_n) + \phi_n(v_n), \qquad n\in \Cc 
\end{align}
where $\phi_n$ and $\psi_n$ are functions solely of the local voltage and reactive power, respectively.
Moreover,  $\phi_n$ and $\psi_n$ are continuously differentiable.
\end{assumption}
Let $L_{\phi_n}$ and $L_{\psi_n}$ represent the Lipschitz constants of $\phi_n$ and $\psi_n$, respectively. For convenience, denote $L_{\phib} := \max_{n \in \Cc} L_{\phi_n}, L_{\psib} := \max_{n \in \Cc} L_{\psi_n}$.
\begin{remark}\longthmtitle{Comparison to existing forms of equilibrium function}\label{rmk:eq-function}
In most of the existing local Volt/Var control schemes of the literature~\cite{XZ-MF-ZL-LC-SHL:21,GC-ZY-MKS-JC:22-cdc}, the equilibrium functions are a special case of Assumption~\ref{ass:differentiability}, 
as they only depend on the local voltages (i.e., $\psi_n = 0$).
Also, $h_n$ could further be generalized as a function that takes all local information $v_n,q_n,p_n$ as arguments. Since the inclusion of $p_n$ does not affect the following closed-loop stability analysis, here we consider $h_n$'s functions of $q_n$ and $v_n$. \oprocend
\end{remark}


\section{Learning Stable Local Volt/Var Controllers}

In this section, we propose a framework for learning local Volt/Var controllers with closed-loop stability guarantees.

\subsection{Local Volt/Var Controllers Design}\label{sec:control_rule}

Here, we propose a local control scheme to steer the system toward configurations meeting~\eqref{eq:equilibria} and provide conditions on the equilibrium functions that ensure asymptotic convergence.
The control algorithm is an \emph{incremental rule} of the form
\begin{equation}\label{eq:bus_react_upd} 
    q_n(t+1) = q_n(t) + \epsilon (h_n(q_n(t),v_n(t)) - q_n(t)), \; n \in \Cc ,
\end{equation} 
with $0 \leq \epsilon \leq 1$. Note that the set $\Qc_n$ is forward invariant under~\eqref{eq:bus_react_upd}, i.e.,  if $q_n(t) \in \Qc_n$, we also have $q_n(t+1)\in \Qc_n$. Let  $\Qc = \times_{n \in \Cc}\Qc_n$ and build the functions $\hbf$, $\phib$, and $\psib$ collecting all the functions $h_n$, $\phi_n$, and $\psi_n$, respectively. 
Then~\eqref{eq:hn} implies that  $\hbf: \Qc \times \real^C \rightarrow \Qc$, and that~\eqref{eq:eq_fun} yields
$$\hbf(\qbf_\Cc,\vbf_\Cc) = \psib(\qbf_\Cc) + \phib(\vbf_\Cc).$$
The power network evolution can then be described by
\begin{subequations} \label{eq:sys_dyn} 
\begin{align}
  &\qbf_\Cc(t + 1) = (1 - \epsilon)\qbf_\Cc(t) + \epsilon \hbf (\qbf_\Cc(t),\vbf_\Cc(t)) , \label{eq:sys_dyn-q} \\
  &\vbf_\Cc(t + 1) = \Xbf \qbf_\Cc(t+1) + \hat \vbf_\Cc . \label{eq:sys_dyn-v}
\end{align} 
\end{subequations}
Plugging~\eqref{eq:sys_dyn-v} into~\eqref{eq:sys_dyn-q}, we obtain the operator $\gbf:\Qc\rightarrow\Qc$ in terms of the reactive power
\begin{align*}
\gbf(\qbf_\Cc):= \qbf_\Cc \!+\! \epsilon(\psib(\qbf_\Cc) \!+\! \phib(\Xbf \qbf_\Cc \!+\! \hat \vbf_\Cc) - \qbf_\Cc).
\end{align*}
Note that the iteration $\qbf_\Cc (t+1) = \gbf (\qbf_\Cc(t))$ precisely corresponds to the control rule~\eqref{eq:bus_react_upd}. Since $\Qc$ is convex, compact and $\gbf$ is continuous, it follows from the Brouwer's fixed-point theorem~\cite{AG-JD:03} that $\gbf$ has a fixed point. In other words, system~\eqref{eq:sys_dyn} admits an equilibrium. Note that any equilibrium 
$(\qbf_\Cc^\sharp,\vbf_\Cc^\sharp)$  of~\eqref{eq:sys_dyn} satisfies by definition
\begin{subequations}\label{eq:fixed-point}
\begin{align}
    \qbf_\Cc^\sharp &= \hbf(\qbf_\Cc^\sharp,\vbf_\Cc^\sharp), \label{eq:fixed-point-q} \\ 
    \vbf_\Cc^\sharp &= \Xbf \qbf_\Cc^\sharp + \hat \vbf_\Cc. \label{eq:fixed-point-v}
\end{align}
\end{subequations}
That is, any equilibrium of~\eqref{eq:sys_dyn} satisfies~\eqref{eq:equilibria}.
The next result provides conditions on the equilibrium functions that guarantee the uniqueness and asymptotic stability of the equilibrium.

\begin{theorem}\longthmtitle{Uniqueness and asymptotic stability of the equilibrium}\label{thm:global-stability}
Under Assumption~\ref{ass:differentiability},
the system~\eqref{eq:sys_dyn} has an unique equilibrium point which is asymptotically stable if
\begin{equation}
L_{\psib} + L_{\phib} \|\Xbf\| < 1 .
\tag{C1}
\label{eq:Lipschtz_cond}
\end{equation}
\end{theorem}
\begin{proof}
 For all $\qbf_\Cc,\qbf_\Cc^\prime \in \Qc$, note that 
 \begin{align*}
     \gbf(\qbf_\Cc)  - \gbf(\qbf_\Cc^\prime)
  &=  (1 - \epsilon)(\qbf_\Cc- \qbf_\Cc^\prime) + \epsilon(\psib(\qbf_\Cc) - \psib(\qbf_\Cc^\prime)) \\
    & \quad + \epsilon(\phib(\Xbf \qbf_\Cc + \hat \vbf_\Cc) - \phib(\Xbf \qbf_\Cc^\prime + \hat \vbf_\Cc))
    \\
    &  = 
((1 - \epsilon)\Ibf + \epsilon \mathbf \Psi + \epsilon \mathbf \Phi \Xbf )( \qbf_\Cc-\qbf_\Cc^\prime) ,
\end{align*}
where $\bf \Psi$ and $\bf \Phi$ are diagonal matrices with diagonal elements
\begin{align*}
&\Psi_n := 
\begin{cases}
\frac{\psi_n(q_n) - \psi_n(q_n')}{q_n - q_n'} & q_n \neq q_n', \\
\hfil 0 & q_n = q_n',
\end{cases}\\
&\Phi_n := 
\begin{cases}
\frac{\phi_n(v_n) - \phi_n(v_n')}{v_n - v_n'} & v_n \neq v_n', \\
\hfil 0 & v_n = v_n' .
\end{cases}
\end{align*}
Note that 
\begin{align*}
    \|\gbf(\qbf_\Cc)  - \gbf(\qbf_\Cc^\prime)\|  & \leq \| (1 - \epsilon) \Ibf + \epsilon \mathbf \Psi + \epsilon \mathbf \Phi \Xbf \|  \|\qbf_\Cc-\qbf_\Cc^\prime\| 
   \\
    &\leq (1 - \epsilon + \epsilon(L_{\psib} + L_{\phib} \|\Xbf \|))  \|\qbf_\Cc-\qbf_\Cc^\prime\|.
\end{align*}
From~\eqref{eq:Lipschtz_cond}, it follows that $0 < 1 - \epsilon + \epsilon(L_{\psib} + L_{\phib} \|\Xbf \|) < 1$, i.e., $\gbf$ is a contraction. Invoking the Banach's fixed-point theorem~\cite{AG-JD:03}, we conclude that system~\eqref{eq:sys_dyn} has an unique equilibrium which is asymptotic stable. 
\end{proof}

Condition \eqref{eq:Lipschtz_cond} allows the equilibrium functions to have arbitrary shapes, unlike the classic monotone piecewise linear form employed in the literature, see e.g.,~\cite{IEEE1547, XZ-MF-ZL-LC-SHL:21}. However, condition~\eqref{eq:Lipschtz_cond} constrains both the slope of the functions $\{\phi_n,\psi_n\}_{n \in \Cc}$, i.e., $L_{\phib} <  1/\|\Xbf\|$ and $L_{\psib} < 1$. The drawback is that the range of $h_n$ might be a strict subset of $\Qc_n$, especially when $q_{\max,n}$ is large, meaning that the corresponding DER $n$ can not fully exploit its generation capabilities.

This limitation can be addressed by  relaxing the slope constraints on the functions $\{\phi_n\}_{n \in \Cc}$ at the cost of requiring them to be \emph{decreasing}, as in e.g.,~\cite{GC-ZY-MKS-JC:22-cdc,ZY-GC-MKS-JC:22-tps,XZ-MF-ZL-LC-SHL:21} (these works, however, consider equilibrium functions of the form $h_n = \phi_n$, i.e., only depending on local voltage magnitudes). The next result provides another set of conditions ensuring the uniqueness and asymptotic stability of the equilibrium of~\eqref{eq:sys_dyn} for equilibrium functions of the form~\eqref{eq:eq_fun}. 
To begin with, we first give a result bounds the eigenvalues of a normal matrix with perturbation that will be used next.

\begin{lemma}\longthmtitle{Bauer and Fike Theorem~\cite[Corollary 6.3.4]{RAH-CRJ:12}}\label{lem:bauer_and_fike}
Let $\Abf,\Ebf \in \real^{n \times n}$ and suppose that $\Abf$ is normal. If $\hat{\lambda}$ is an eigenvalue
of $\Abf + \Ebf$, then there is an eigenvalue $\lambda$ of $\Abf$ such that $|\hat{\lambda} - \lambda| \leq \|\Ebf\|$.
\end{lemma}

Now we are ready to give conditions that ensure asymptotic stability of the closed-loop system.

\begin{theorem}\longthmtitle{Uniqueness and asymptotic stability of the equilibrium}\label{thm:local-stability}
Under Assumptions~\ref{ass:differentiability}, the system~\eqref{eq:sys_dyn} has an unique equilibrium point if
\begin{equation}
\tag{C2}
\label{eq:Lpsi}
\{\phi_n\}_{n \in \Cc}~\text{are decreasing},\quad L_{\psib} < 1.    
\end{equation}
Further, the equilibrium point is asymptotically stable if
\begin{equation}
\label{eq:eps_cond}
\epsilon < \frac{2}{L_{\psib} + L_{\phib} \|\Xbf\| + 1}.    
\end{equation}
\end{theorem}

\begin{proof}
The uniqueness of the equilibrium can be proved by contradiction. Assume that $\qbf_C$ and $\qbf_C'$ are equilibrium points, with $\qbf_C \neq \qbf_C'$. Then,
$$\qbf_\Cc - \qbf_\Cc' = \phib(\vbf_\Cc) + \psib(\qbf_\Cc)  - \phib(\vbf_\Cc') - \psib(\qbf_\Cc').$$
It follows that
\begin{align*}
  (\Ibf - \mathbf \Psi) (\qbf_\Cc - \qbf_\Cc') = \phib(\vbf_\Cc)  - \phib(\vbf_\Cc') = \mathbf \Phi (\vbf_\Cc - \vbf_\Cc'),
\end{align*}
and consequently
\begin{align}\label{eq:unique-eq-proof-A}
     \qbf_\Cc - \qbf_\Cc' &= (\Ibf - \mathbf \Psi)^{-1}  \mathbf \Phi(\vbf_\Cc - \vbf_\Cc').
\end{align}
Note that condition~\eqref{eq:Lpsi} and the fact that the $\{\phi_n\}_{n \in \Cc}$ are decreasing yields $\Ibf - \mathbf \Psi \succ 0$ and $\mathbf \Phi \prec 0$ and hence $(\Ibf - \mathbf \Psi)^{-1}  \mathbf \Phi \prec 0$. 
On the other hand,~\eqref{eq:fixed-point-v} yields $\qbf_\Cc - \qbf_\Cc' = \Xbf^{-1} (\vbf_\Cc - \vbf_\Cc')$,
which contradicts~\eqref{eq:unique-eq-proof-A} as $\Xbf^{-1} \succ 0$. 

Next, we establish asymptotic stability. The Jacobian of $\gbf$~is:
$$\Jbf_{\gbf} := (1-\epsilon)\Ibf + \epsilon \Jbf_{\psib} + \epsilon \Jbf_{\phib} \Xbf,$$
where the diagonal matrices $\Jbf_{\phib}$ and $\Jbf_{\psib}$ are the Jacobian matrices of $\phib$ and $\psib$, respectively. Since $\{\phi_n\}_{n \in \Cc}$ are decreasing, $\Jbf_{\phib} \prec 0$. Notice that $\Jbf_{\gbf}$ is similar to the matrix 
\begin{align*}
\hat \Jbf_{\gbf} &= (1-\epsilon)\Ibf + \epsilon \Jbf_{\psib} - \epsilon (-\Jbf_{\phib})^{\frac 1 2} \Xbf (-\Jbf_{\phib})^{\frac 1 2} \\
& = \Ibf + \epsilon (\Jbf_{\psib} - \Ibf) - \epsilon (-\Jbf_{\phib})^{\frac 1 2} \Xbf (-\Jbf_{\phib})^{\frac 1 2},
\end{align*}
which is symmetric and therefore its eigenvalues are all real. Hence, $\Jbf_{\gbf}$ and $\hat \Jbf_{\gbf}$ share the same real eigenvalues. We first show $\lambda_{\max}(\Jbf_{\gbf}) < 1$. Since $L_{\psib} < 1$, we have $\Jbf_{\psib} - \Ibf \prec 0$; combined with the fact that $(-\Jbf_{\phib})^{\frac 1 2} \Xbf (-\Jbf_{\phib})^{\frac 1 2} \succ 0$, we conclude that $\lambda_{\max}(\hat \Jbf_{\gbf}) < 1$, and therefore $\lambda_{\max}(\Jbf_{\gbf}) < 1$. 

Next, we show that $\lambda_{\min}(\Jbf_\gbf) > -1$. Since $(1-\epsilon)\Ibf + \epsilon \Jbf_{\psib}$ is symmetric, using Lemma~\ref{lem:bauer_and_fike},
\begin{align*}
   \lambda_{\min}(\Jbf_{\gbf}) & \geq 1 - \epsilon + \lambda_{\min}(\epsilon \Jbf_{\psib}) - \|\epsilon \Jbf_{\phib} \Xbf\| \\
   & \geq 1 - \epsilon - \epsilon \lambda_{\max}(-\Jbf_{\psib}) - \epsilon \|\Jbf_{\phib} \Xbf\| \\
   & > 1 - \frac{2(\|\Jbf_{\psib}\| + \|\Jbf_{\phib} \Xbf\| + 1)}{L_{\psib} + L_{\phib} \|\Xbf\| + 1} \geq -1.
\end{align*}
This ensures the asymptotic stability of the equilibrium~\cite[Theorem 3.3]{NB-RC-LS:18}, completing the proof.
\end{proof}

\begin{remark}\longthmtitle{Trade-offs between coupled voltage-power slope constraint and  reactive power slope constraint}\label{rmk:trade-off}
We envision the \emph{coupled voltage-power slope constraint} (CVP-SC), cf.~\eqref{eq:Lipschtz_cond}, to provide meaningful designs for the case of a DG with relatively small size generators, for which it does not quite limit the generation usage, and the more flexible shape of equilibrium functions could enhance optimality. Instead, we believe the \emph{reactive power slope constraint} (RP-SC), cf.~\eqref{eq:Lpsi}, is more suitable for the case of a DG with relatively big size generators, since the lack of slope limitations on the functions $\{\phi_n\}_{n \in \Cc}$ could help the generators to make full use of their reactive power compensation capabilities and thus leads to better performance. We illustrate such trade-offs in the simulations through two different DG cases.
Finally, we note that CVP-SC ensures \emph{global} asymptotic stability, whereas RP-SC ensures \emph{local} asymptotic stability, as the latter relies on the linearization of the operator $\gbf$ at equilibrium points. Moreover, CVP-SC allows arbitrary $\epsilon \in [0,1]$, while RP-SC might be more restrictive on the selection of $\epsilon$, cf.~\eqref{eq:eps_cond}.
\oprocend
\end{remark}

\subsection{Learning The Equilibrium Functions}
The learning process consists of the following steps.
First, 
we build a set $\{(\pbf^k,\qbf_\Lc^k,\qbf_\Cc^k)\}_{k=1}^K$ of $K$ load-generation scenarios. 
One can  obtain these scenarios via random sampling from assumed probability distributions, historical data, or from forecasted conditions for a look-ahead period. 
Second, we solve the power flow equation~\eqref{eq:v=Rp+Xq} and ORPF problem~\eqref{eq:ORPF} for these $K$ scenarios to obtain
$\vbf_\Cc^k(\pbf^k,\qbf_\Lc^k,\qbf_\Cc^k)$
and $\qbf_\Cc^{\star,k}(\pbf^k,\qbf_\Lc^k)$, respectively.
Finally, we build a dataset for each $n \in \Cc$ of the form $\Dc_n=\{(v_{n}^k,q_{n}^k,q^{\star,k}_{n})\}_{k=1}^K$,
and each equilibrium function $h_n$ is then learned by solving 
\begin{align}\label{eq:learning_problem}
    \min_{h_n}\ & ~\sum_{k=1}^K |q_{n}^{\star,k} - h_n(v_n^k,q_n^k) |^2  \\
    \mathrm{s.t.}\  & ~h_n~\text{designed under CVP-SC or~RP-SC}. \notag
\end{align}
Note that solving~\eqref{eq:learning_problem} requires parameterizing $\{\phi_n\}_{n \in \Cc}$ and $\{\psi_n\}_{n \in \Cc}$ either to be decreasing or slope-limited. We adopt the single-hidden-layer neural network approximation method in~\cite[Section IV]{ZY-GC-MKS-JC:22-tps}, which provides convenient conditions on the weights of neural networks to ensure monotonicity and slope limitations.
Then,~\eqref{eq:learning_problem} can be solved using suitable renditions of (stochastic) gradient descent prevalent for neural network training, e.g., the Adam algorithm~\cite{DPK-JB:15}.


\section{Case Study}
\label{sec:sim}

We validate our approach on a modified version of the IEEE 37-bus feeder taken from~\cite{GC-ZY-MKS-JC:22-cdc}, reported in Fig.~\ref{fig:ieee37}.
We benchmark the resulting control rules against our previous work~\cite{ZY-GC-MKS-JC:22-tps}, where the equilibrium functions solely depend on local voltages, which showed 
significant  enhancements in the optimality gap with respect to (optimized) linear droop control. 

\begin{figure}[t]
\centering	
\includegraphics[width=0.5\columnwidth]{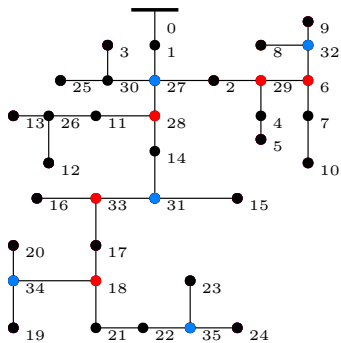}
\caption{The modified IEEE 37-bus feeder. Blue nodes and red nodes represent generators in $\Cc_1$ and $\Cc_2$, respectively.}
\vspace*{-.5em}
\label{fig:ieee37}
\end{figure}

\subsubsection*{Simulation Setup} 
We consider two cases. In \texttt{Case-1}, there are 5 generators, placed at buses $\Cc_1 =\{27, 31, 32, 34, 35\}$,  with generation capability $\qbf_{\max} = 0.4 \times \ones$ MVAR; in \texttt{Case-2}, there are 10 generators, placed on buses $\Cc_1 \cup \Cc_2$, where $\Cc_2 = \{6,18,28,29,33\}$, with generation capability $\qbf_{\max} = 0.2 \times \ones$ MVAR. 
In both cases, $\qbf_{\min} = -\qbf_{\max}$, $\vbf_{\max} = 1.05 \times \ones$ p.u., and $\vbf_{\min} = 0.95 \times \ones$ p.u.
Notice that in \texttt{Case-2}, we have a bigger number of smaller generators than in \texttt{Case-1}.
For our experiments, we use the same minute-based load and solar generation data of~\cite{ZY-GC-MKS-JC:22-tps}.
Moreover, we randomly generate five reactive power injection $\qbf_\Cc$ from $[\qbf_{\min},\qbf_{\max}]$ for each minute-based data, resulting in a total of $K = 1440 \times 5 = 7200$ load-generation profiles.
We use the CVX toolbox~\cite{MG-SB:14-cvx} to solve the power flow equation~\eqref{eq:v=Rp+Xq} as well as the ORPF problem~\eqref{eq:ORPF} for all load-generation profiles. We train the neural networks to solve~\eqref{eq:learning_problem} using TensorFlow 2.7.0 in Google Colab with a single TPU with 32 GB memory. The training hyperparameters are the same with~\cite{ZY-GC-MKS-JC:22-tps}. 

\subsubsection*{Learning Performance}
The equilibrium functions are computed by solving~\eqref{eq:learning_problem}.
Following our discussions in Remark~\ref{rmk:trade-off}, we use RP-SC in \texttt{Case-1} since it is more suitable for the case with relatively big size generators, and use CVP-SC in \texttt{Case-2} due to its advantages for the case with relatively small-size generators. In all cases, we use~\cite{ZY-GC-MKS-JC:22-tps} as the baseline.
Table~\ref{tab:loss} illustrates the learning performance for the two test cases. In both, the algorithms developed under CVP-SC in Theorem~\ref{thm:global-stability} and RP-SC in Theorem~\ref{thm:local-stability}, when
compared with the one of~\cite{ZY-GC-MKS-JC:22-tps}, achieve lower training loss, computed as 
$\frac{1}{KC} \sum_{n \in \Cc} \sum_{k=1}^K |q_{n}^{\star,k} - h_n(v_n^k,q_n^k) |^2$.
This shows that the inclusion of reactive power as argument of the equilibrium function helps increase the prediction accuracy. Figs.~\ref{fig:func-big-gen} and~\ref{fig:func-small-gen} plot the learned functions $\phi_{35}$ and $\psi_{35}$ for \texttt{Case-1} and \texttt{Case-2}, respectively. In \texttt{Case-1}, although CVP-SC allows the function $\phi_n$ to have arbitrary shape, the learned function $\phi_{35}$ in Fig.~\ref{fig:func-big-gen} is decreasing. Thus, the more restrictive slope limitations of CVP-SC make its performance worse than RP-SC in approximating the ORPF solutions, which explains CVP-SC yielding greater training loss than RP-SC for \texttt{Case-1} in Table~\ref{tab:loss}. 
Instead, in \texttt{Case-2}, although CVP-SC has more restrictive slope limitations on the functions $\{\phi_n\}_{n \in \Cc}$, it does not affect much the performance since the generation capability is relatively small. Instead, the monotonicity requirement for  the functions $\{\phi_n\}_{n\in\Cc}$ in RP-SC degrades the prediction accuracy as one can see that the learned function $\phi_{35}$ in Fig.~\ref{fig:func-small-gen} for CVP-SC is not always decreasing. This is the reason that CVP-SC works better for \texttt{Case-2}, as Table~\ref{tab:loss} suggests. These observations are consistent with our discussion in Remark~\ref{rmk:trade-off}.

\begin{table}[tb]
\centering
\caption{Average training loss and optimality gap distance between actual reactive power setpoints and ORPF solutions\tablefootnote{The method in our previous work~\cite{ZY-GC-MKS-JC:22-tps} is a special case of RP-SC with $\psi_n = 0$ for all $n \in \Cc$.}}
\vspace{-1ex}
\begin{tabular}{ c | c c | c c}
\toprule
& \multicolumn{2}{c|}{Training loss} & \multicolumn{2}{c}{Average distance} \\
\midrule
 & Baseline & Improvement & Baseline  & Improvement \\
\midrule
\texttt{Case-1}  & $6.3 \times 10^{-3}$ & $39.7 \%$ & 0.1969 & $13.09 \%$\\
\texttt{Case-2}  & $7.9 \times 10^{-4}$ & $14.0 \%$ & 0.1235 & $39.41 \%$\\
\bottomrule
\end{tabular}
\label{tab:loss}
\end{table}

\begin{figure}[ht]
    \centering
    \subfigure[CVP-SC]{
    \includegraphics[width=.25\textwidth]{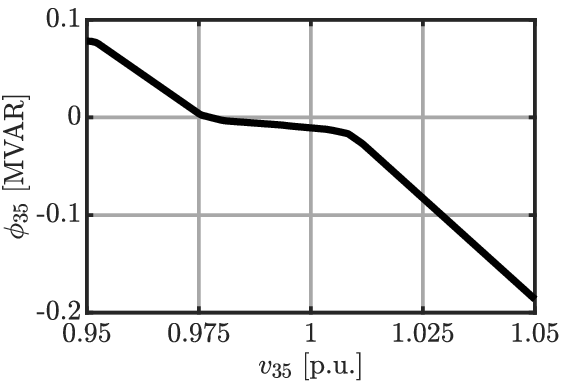}
    \hspace{-.15in}
    \includegraphics[width=.25\textwidth]{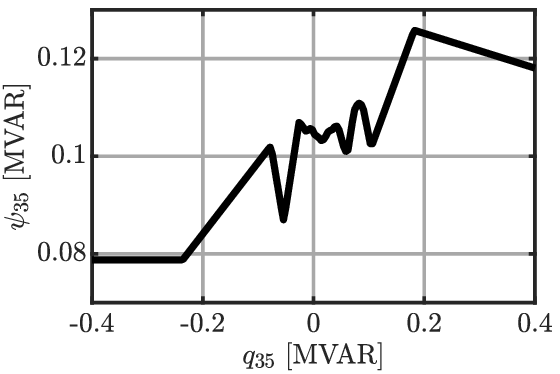}
    \label{fig:global-big-gen}
    }\\
    \subfigure[RP-SC]{
    \includegraphics[width=.25\textwidth]{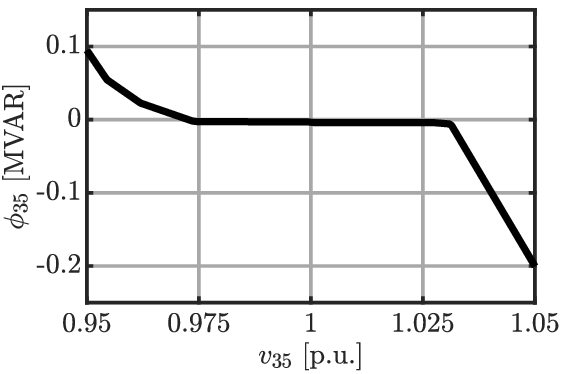}
    \hspace{-.15in}
    \includegraphics[width=.25\textwidth]{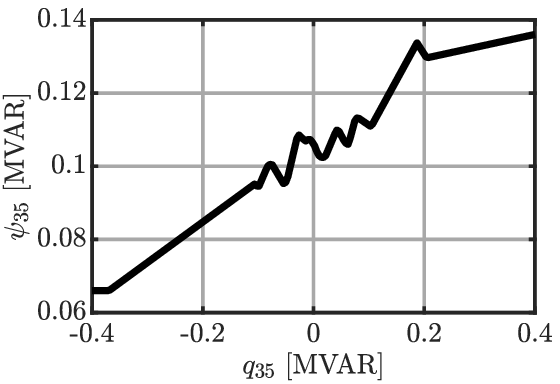}
    \label{fig:local-big-gen}
    }
    \caption{Leaned functions $\phi_{35}$ and $\psi_{35}$ of node 35 for \texttt{Case-1} under (a) CVP-SC and (b) RP-SC.}
    \label{fig:func-big-gen}
\end{figure}

\begin{figure}[ht]
    \centering
    \subfigure[CVP-SC]{
    \includegraphics[width=.25\textwidth]{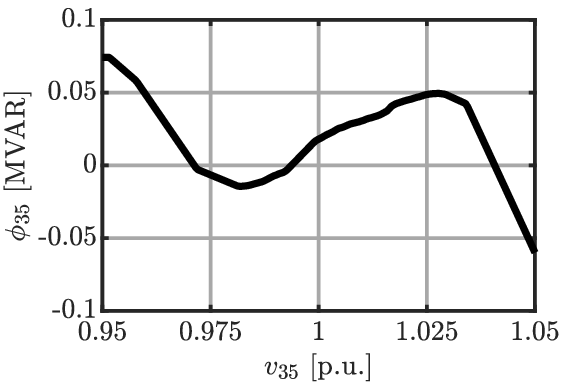}
    \hspace{-.15in}
    \includegraphics[width=.25\textwidth]{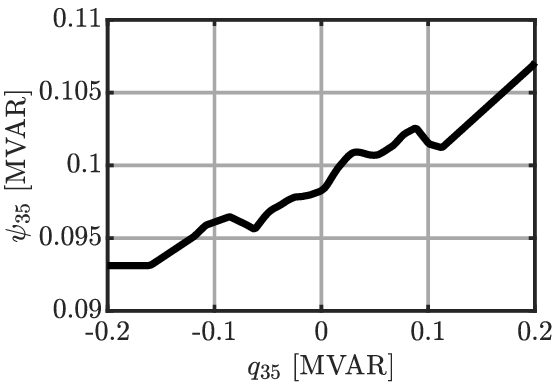}
    \label{fig:global-small-gen}
    }\\
    \subfigure[RP-SC]{
    \includegraphics[width=.25\textwidth]{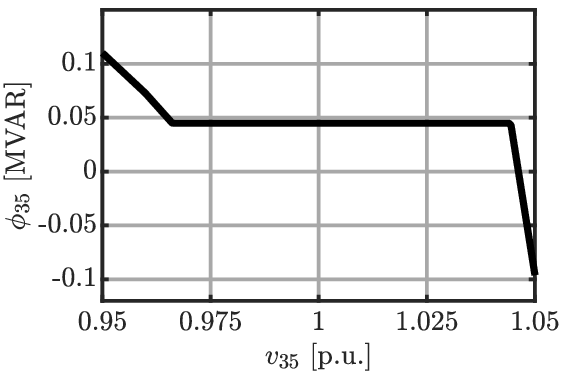}
    \hspace{-.15in}
    \includegraphics[width=.25\textwidth]{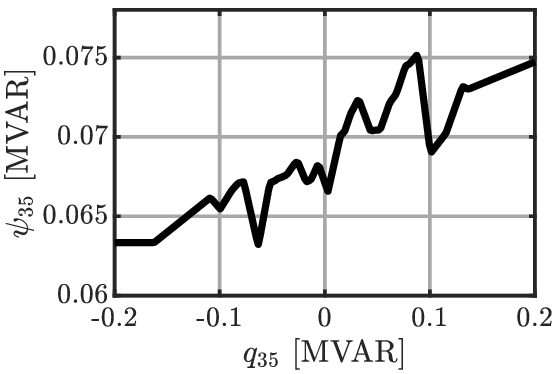}
    \label{fig:local-small-gen}
    }
    \caption{Leaned functions $\phi_{35}$ and $\psi_{35}$ of node 35 for \texttt{Case-2} under (a) CVP-SC and (b) RP-SC.}
    \label{fig:func-small-gen}
\end{figure}

\subsubsection*{Control Performance}
We test the control performance of the proposed methods under CVP-SC and RP-SC. 
Although our stability analysis is done for the linearized power flow model, here we employ \textsc{Matpower}~\cite{RDQ-CEMS-RJT:11} to solve the AC power flow to run the simulations.
Figs.~\ref{fig:react-big-gen} and~\ref{fig:react-small-gen}
report the evolution of the DERs’ reactive power injections for both \texttt{Case-1} and \texttt{Case-2} to show the stability of the control algorithms.
Finally, we test the proposed methods in a scenario where load-generation profiles are time-varying. Specifically, we randomly perturb (5\%) the consumption data from 15:00 to 17:00 to obtain the testing load-generation profiles. 
We set $\epsilon = 0.1$ and consider 120 iterations of~\eqref{eq:bus_react_upd} using the controllers developed under CVP-SC and RP-SC.
Table~\ref{tab:loss} summarizes the average distances between the actual reactive power setpoints and the ORPF solutions, i.e., $\| \qbf_\Cc - \qbf_\Cc^{\star} \|$. The performance displayed here by CVP-SC and RP-SC illustrates their respective advantages in different DG cases and their significant  improvement compared to the baseline~\cite{ZY-GC-MKS-JC:22-tps}.
Notably, we observe that the performance improvement achieved by CVP-SC in \texttt{Case-2} is greater than that by RP-SC in \texttt{Case-1}. This is because in \texttt{Case-2} CVP-SC enjoys the performance improvement  from the inclusion of reactive power as an argument of the equilibrium function as well as the more flexible shape of equilibrium function. Instead, in \texttt{Case-1}, the performance improvement achieved by RP-SC is only due to the inclusion of reactive power as argument of the equilibrium function.

\begin{figure}[ht]
    \centering
    \subfigure[CVP-SC]{
    \includegraphics[width=.245\textwidth]{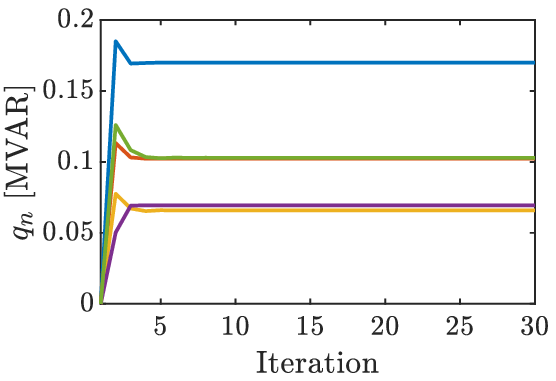}
    \label{fig:react-global-big-gen}
    }
    \hspace{-.3in}
    \subfigure[RP-SC]{
    \includegraphics[width=.245\textwidth]{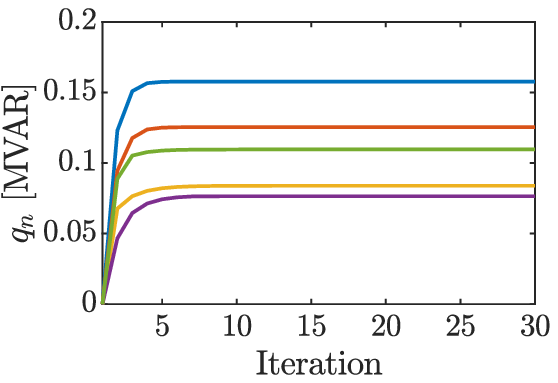}
    \label{fig:react-local-big-gen}
    }
    \caption{Evolution of reactive power setpoints for \texttt{Case-1} with 30 iterations of~\eqref{eq:bus_react_upd} using load-generation profiles of 1095-th minute. For CVP-SC, $\epsilon$ is chosen as 1, and for RP-SC, $\epsilon$ is chosen as $0.79$, which satisfies $\epsilon < \frac{2}{L_{\psib} + L_{\phib}\|\Xbf\| + 1} = 0.7916$ in Theorem~\ref{thm:local-stability}.}
    \label{fig:react-big-gen}
    \vspace*{-1.5ex}
\end{figure}
\begin{figure}[ht]
    \centering
    \subfigure[CVP-SC]{
    \includegraphics[width=.245\textwidth]{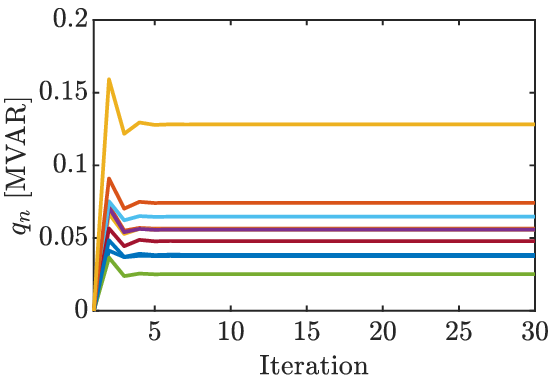}
    \label{fig:react-global-small-gen}
    }
    \hspace{-.3in}
    \subfigure[RP-SC]{
    \includegraphics[width=.245\textwidth]{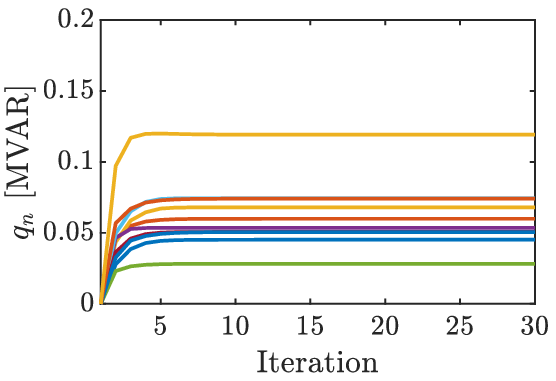}
    \label{fig:react-local-small-gen}
    }
    \caption{Evolution of reactive power setpoints for \texttt{Case-2} with 30 iterations of~\eqref{eq:bus_react_upd} using load-generation profiles of 1095-th minute. For CVP-SC, $\epsilon$ is chosen as 1, and for RP-SC, $\epsilon$ is chosen as $0.64$, which satisfies $\epsilon < \frac{2}{L_{\psib} + L_{\phib}\|\Xbf\| + 1} = 0.6471$ in Theorem~\ref{thm:local-stability}.}
    \label{fig:react-small-gen}
        \vspace*{-1.5ex}
\end{figure}

\section{Conclusions}
We have developed a learning method for synthesizing provably stable local Volt/Var controllers for efficient network operation of distribution grids. We proposed an incremental control algorithm steering the network towards configurations defined by functions termed equilibrium functions, which depend both on local voltages and local reactive powers. We identified two sets of slope constraint conditions on the equilibrium functions to ensure the stability of the algorithm. The theoretical analysis and simulation results illustrate the trade-offs between the two types of conditions, as the reactive  power slope constraint  is better in DGs with relatively large-size generators and coupled voltage-power slop constraint is more suitable in DGs with relatively small-size generators. 
Future work will explore more general forms of the equilibrium functions, relax the assumptions on its components, and extend our analysis to Lipschitz equilibrium functions.

\bibliographystyle{ieeetr}

\end{document}